\documentclass[runningheads,a4paper]{article}

\usepackage{amsmath,amsfonts,amssymb,amsthm,graphicx,graphics,epsf}
\usepackage{geometry}

\usepackage{color,soul}

\usepackage{algorithmic}
\usepackage{algorithm}

\usepackage{hyperref}

\usepackage{graphicx}

\usepackage{url}

\newtheorem{theorem}{Theorem}[section]
\newtheorem{corollary}[theorem]{Corollary}
\newtheorem{lemma}[theorem]{Lemma}
\newtheorem{proposition}[theorem]{Proposition}
\newtheorem{observation}[theorem]{Observation}

\newcommand{\sep}{\ensuremath{\mathbf{c'}}} 
\newcommand{\Sep}{\ensuremath{\mathbf{C'}}}
\newcommand{\CP}{\ensuremath{C_P}}  
\newcommand{\cp}{\ensuremath{c_P}}  
\newcommand{\tang}{\ensuremath{\mathbf{\phi'}}}

\urldef{\mailsa}\path|{alfred.hofmann, ursula.barth, ingrid.haas, frank.holzwarth,|
\urldef{\mailsb}\path|anna.kramer, leonie.kunz, christine.reiss, nicole.sator,|
\urldef{\mailsc}\path|erika.siebert-cole, peter.strasser, lncs}@springer.com|

\def\hull#1{{\text{CH}(#1)}}

\begin{document}

\title{Circle separability queries in logarithmic time}

\author{Greg Aloupis $^*$ \and Luis Barba \thanks{D\'epartement d'Informatique, Universit\'e Libre de Bruxelles, Brussels, Belgium}\and Stefan Langerman $^*$
}
\date{}

\maketitle

\begin{abstract}
Let $P$ be a set of $n$ points in the plane.
In this paper we study a new variant of the circular separability problem in which
a point set $P$ is preprocessed so that one can quickly answer queries of the following form:
Given a  geometric object $Q$, 
report the minimum circle containing $P$ and exluding $Q$.
Our data structure can be constructed in $O(n\log n)$ time using $O(n)$ space, 
and can be used to answer the query when $Q$ is either a circle or a convex $m$-gon in $O(\log n)$ or $O(\log n + \log m)$ time, respectively.
\end{abstract}

\section{Introduction}

The planar separability problem consists of constructing, if possible, a boundary that  separates the plane into two components such that two given sets of geometric objects become isolated.
Typically this boundary is a single curve such as a line, circle or simple polygon, meaning that each component of the plane is connected.
Often there is an additional objective of minimizing some feature of the boundary, for instance the radius of a circle. 

Probably the most classic instance of this problem is to separate two given point sets with a circle (or a line, which is equivalent to an infinitely large circle).  
This problem was developed because of the several applications it has in pattern recognition and image processing~\cite{SeparatingPointsByCricles,DigitalDiskAndCompacness}.
Circle separability has been extensively studied and usually two variants are considered: The decision problem, where the existence of the separating circle is tested, and the Min-Max problem, where the minimum or maximum separating circle is to be determined.

A separating line can be found, if it exists, using linear programing. In the plane this takes linear time by Megiddo's algorithm~\cite{LinearProgrammingInLinearTime}.
For circle separability, O'Rourke, Kosaraju and Megiddo~\cite{ComputingCircularSeparability} gave a linear-time algorithm for the decision problem (in fact spherical separability in any fixed dimension), improving earlier bounds~\cite{ComputingCircularSeparabilityOfPlanarPointSets,DigitalDiskAndCompacness}.
They also gave an $O(n\log n)$ algorithm for finding the largest separating circle and a linear time algorithm for finding the minimum separating circle between any two finite point sets.
Extending these ideas, Boissonat et al.~\cite{boissonat-etal} gave a linear-time algorithm to report the smallest separating circle for two simple polygons, if it exists.

The separability problem has also been studied when two point sets are to be isolated by the boundary of a simple polygon.
Edelsbrunner and Preparata~\cite{MinimumPolygonalSeparation} gave an $O(n\log n)$-time algorithm to find a separating convex polygon with minimum number of vertices. 
This was shown to be optimal if the optimal separator has linear size.  
They also gave a $O(kn)$ algorithm to find a convex polygonal separator with either $k$ or $k{+}1$ vertices, where $k$ is the size of the optimal solution.

Aggarwal et al.~\cite{Aggarwal:1985} gave an $O(n\log k)$-time algorithm to find the separating (convex) polygon with fewest vertices, between  two nested convex polygons. Again, $k$ is the size of the optimal separator.

Das and Joseph extended the problem to higher dimensions and 
proved that the problem of computing a separating polyhedron, 
having the minimum number of faces, 
for two nested convex polyhedra is NP-complete~\cite{DasJoseph:1990}.
Suri and O'Rourke~\cite{SO} gave a quadratic time algorithm to find a polygonal separator of two simple, not necessarily convex, nested polygons. 
This was improved to $O(n\log n)$ time by Wang~\cite{Wang}.
Finally, Wang and Chan~\cite{WangChan} gave an $O(n\log n)$ algorithm to find a minimal polygonal separator between two arbitrary, not necessarily nested, simple polygons.

In some situations, the given geometric objects and the separating boundary are constrained to lie within some subset of the plane, such as a simple polygon.   
For instance, Demaine et al.~\cite{8people} studied the separation of point sets by chords and geodesic paths inside polygons.

The online version of the separability problems, in which a preprocess is allowed to answer separability queries, is of special interest when a lot of queries are to be processed. 
This is the case in areas like geometric modeling involving objects in motion, where collision detection or largest empty space recognition queries are extensively used. Therefore, the online variants of some separability problems have been studied as well:
Augustine et al~\cite{augustine-etal} show how to preprocess a point set $P$,
so that the largest circle, isolating $P$ from a query point, can be found in logarithmic time.
They also obtain this query time when $P$ represents the boundary of a simple polygon.

For the line separability problem, Edelsbrunner show that a point set $P$ can be preprocessed in $O(n\log n)$ time, 
so that a separating line between $P$ and a query convex $m$-gon $Q$ can be computed in $O(\log n + \log m)$ time~\cite{ComputingExtremeDistancesBetweenConvexPolygons}.
In 3D space, Dobkin and Kirkpatrick show that two convex polyhedra can be preprocessed in linear time, so that a separating plane, if any exists, can be computed in $O(\log n\cdot \log m)$  time, given any orientation and position of both polyhedra. In this case, $n$ and $m$ represent the size of each polyhedra.

In this paper we show that a point set $P$ on $n$ points can be preprocessed in $O(n\log n)$ time, using $O(n)$ space, so that for any given convex $m$-gon $Q$, we can find the smallest circle enclosing $P$ and excluding $Q$ in $O(\log n+\log m)$ time.
This improves the $O(\log n\cdot\log m)$ bound presented in~\cite{DynamicCircleSeparability}, which is described in this paper as well.

\section{Preliminaries}

Let $P$ be a set of $n$ points in the plane and let $Q$ be a convex $m$-gon. 
We say that every circle containing $P$ is a {\em $P$-circle}.
We also say that a circle $C$ separates $P$ from $Q$, or simply that $C$ is a \emph{separating circle}, 
if $C$ is a $P$-circle and its interior does not intersect $Q$. 
Likewise, a \emph{separating line} is a straight line leaving the interiors of $P$ and $Q$ in different halfplanes.

Let $\Sep$ denote the minimum separating circle and let $\sep$ be its center.
Note that $\Sep$ always passes through at least two points of $P$, 
since otherwise a smaller separating circle can always be found.
In fact  $\sep$ must lie on an edge of the farthest-point Voronoi diagram  $\mathcal{V}(P)$, which is a tree with leaves at infinity~\cite{Computational geometry: Algorithms and applications}.
For each point $p$ of $P$, let $R(p)$ be 
the farthest-point Voronoi region associated with $p$.

Let $\CP$ be the minimum enclosing circle of $P$. 
If  $\CP$ is constrained by three points of $P$ then 
its center, $\cp$, is at a vertex 
of $\mathcal{V}(P)$.  Otherwise $\CP$ is constrained by exactly two points of $P$ (forming its diameter),
in which case $\cp$ is on the interior of  an edge of $\mathcal{V}(P)$. If that is the case, we will insert $\cp$ into $\mathcal{V}(P)$ by splitting the edge where it belongs.
Thus, we can think of $\mathcal{V}(P)$ as a rooted tree on $\cp$.
For any given point $x$ on $\mathcal{V}(P)$ there is a unique path along $\mathcal{V}(P)$ joining  $\cp$ with $x$, throughout this paper we will denote this path by $\pi_x$.

Given any point $y$ in the plane, 
let $C(y)$ be the minimum $P$-circle with center on $y$
and let $\rho(y)$ be the radius of $C(y)$.
Note that if $y$ belongs to $R(p)$ for some point $p$ of $P$, 
then $\rho(y)$ is given by $d(y,p)$, where $d(\cdot,\cdot)$ denotes the Euclidian distance between any two geometric objects.
Finally, we say that $y$ is a \emph{separating point} if $C(y)$ is a separating circle.

\section{Properties of the minimum separating circle}\label{section:Resultados}
In this section we  describe some properties of $\Sep$,
and the relationship between $\sep$ 
and the farthest-point Voronoi diagram.
These properties are not new. In fact most of the results in this section are either proved, stated, or assumed in~\cite{DynamicCircleSeparability}.

Let $\hull{P}$ denote the convex hull of $P$.
We assume that  the interiors of $Q$ and $\hull{P}$ are disjoint, 
otherwise there is no separating circle.
Also, if $Q$ and $\CP$ have disjoint interiors, 
then $\CP$ is trivially the minimum separating circle.

The following useful property of the farthest-point Voronoi diagram was mentioned
in~\cite{DynamicCircleSeparability}.   We provide a short proof.

\begin{proposition}\label{Rho is monotonic}
Let $x$ be a point on $\mathcal{V}(P)$. The function $\rho$ is monotonically 
increasing along the path $\pi_{x}$
starting at $\cp$.
\end{proposition}
\begin{proof}
Let $u,v$ be two vertices of $\mathcal{V}(P)$.
In~\cite{ConstrainedMinimumEnclosingCircleWithCenterOnAQueryLineSegment}, they proved that if $u$ is an ancestor of $v$, then $\rho(u)< \rho(v)$. 
It only remains to prove that $\rho$ is also monotonically increasing along every edge of $\pi_x$.
Let $e$ be an edge of $\pi_x$ contained in the bisector of the points $p,p'$ of $P$, 
and consider a point $y$ moving along $e$. 
Note that as $d=d(y,p)=d(y,p')$ increases, 
the radius of $C(y)$ also increases since $y$ lies on $R(p)\cap R(p')$.
Thus, by moving the point $y$ we obtain that $\rho(y)$ increases monotonically on every edge along $\pi_x$.
\end{proof}

\begin{observation}\label{Relation P-circle separating circle}
Every $P$-circle contained in a separating circle is also a separating circle.
\end{observation}

The following is stated in~\cite{DynamicCircleSeparability}.  We provide a brief proof.
\begin{proposition}\label{CirculosEnSegmento}
Let $x$ and $y$ be two points in $\mathbb{R}^2$. If $z$ is a point contained in the segment $[x,y]$, 
then $C(z)\subseteq C(x)\cup C(y)$.
\end{proposition}
\begin{proof}
Let $a$ and $b$ be the two points of intersection between  $C(x)$ and $C(y)$. 
Let $r=d(a,z)=d(b,z)$ and let $C_r(z)$ be the circle with center on $z$ and radius $r$; see Figure~\ref{fig:CirculosEnSegmento}.

It is clear that $C(x)\cap C(y)\subseteq C_r(z)$ and since $P\subset C(x)\cap C(y)$,  we infer that $C_r(z)$ is a $P$-circle, therefore $C(z)\subseteq C_r(z)$. 
Furthermore, since $C_r(z)\subseteq C(x)\cup C(y)$ by construction, we conclude that $C(z)\subseteq C(x)\cup C(y)$.
\end{proof}

Observation~\ref{Relation P-circle separating circle} and Proposition~\ref{CirculosEnSegmento} imply that if $C(x)$ and $C(y)$ are both separating circles, then for every $z\in[x,y]$, $C(z)$ is also a separating circle. 
This comes from the fact that $C(z)$ is contained in $C(x)\cup C(y)$ and $Q$ lies outside of this union.
Furthermore, this implies that the minimum separating circle is unique.
Assume otherwise that $C$ and $C'$ are both minimum separating circles with centers $c$ and $c'$, respectively. Thus, for every point $z$ in the open segment $(c,c')$, $C(z)$ is also a separating circle contained in $C\cup C'$. This means that $C(z)$ has an smaller radius than $C$ and $C'$ which would be a contradiction.

\begin{figure}[h]
\begin{center}
\includegraphics[width=85mm]{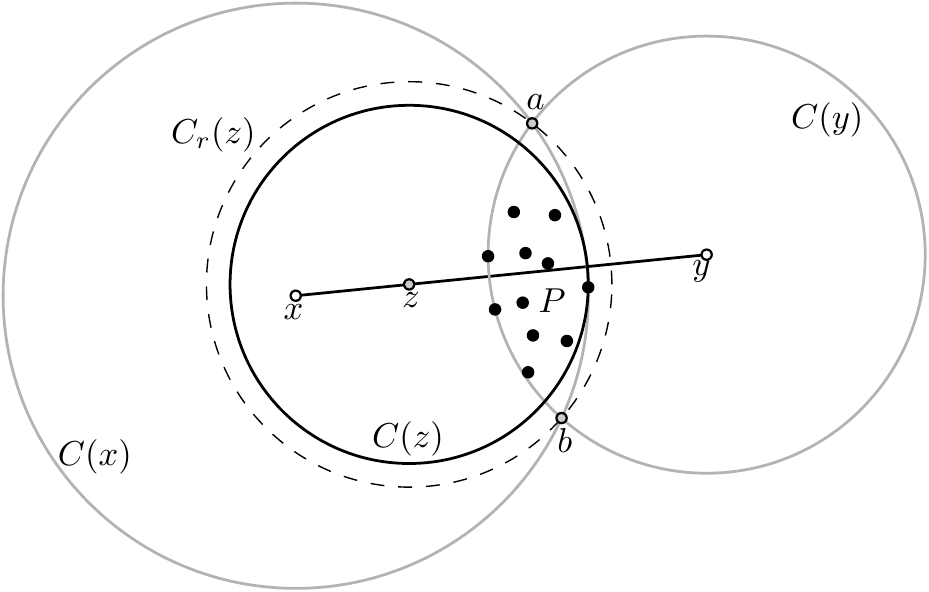}
\caption{\small For every point $z\in [x,y]$ the circle $C(z)$ lies in the union of $C(x)$ and $C(y)$.} 
\label{fig:CirculosEnSegmento}
\end{center}
\end{figure}

The following is also demonstrated in~\cite{DynamicCircleSeparability}, with a reworded proof.
\begin{lemma}\label{Common Ancestor}
Let $x$ and $y$ be two separating points on $\mathcal{V}(P)$.
If $z$ is the lowest common ancestor of $x$ and $y$ in the rooted tree $\mathcal{V}(P)$, then $C(z)$ is a separating circle; moreover $\rho(z) \leq \min\{\rho(x), \rho(y)\}$.
\end{lemma}
\begin{proof}
Since $z$ is a common ancestor of $x$ and $y$, we infer from Proposition~\ref{Rho is monotonic} that $\rho(z)\leq \min\{\rho(x), \rho(y)\}$.   
What remains is to prove that $C(z)$ is a separating circle.
Suppose that $y\notin \pi_x$ and $x\notin \pi_y$, otherwise the result follows 
trivially since $y$ or $x$ will coincide with $z$.
So, let $e_x$ ({\em resp.} $e_y$) be the edge incident to $z$, on the path joining $z$ with $x$ ({\em resp.} $y$).
Consider the radial order of the edges incident to $z$ lying between $e_x$ and $e_y$. 
Let $e'$ be the edge consecutive to $e_x$ in this order.
Therefore, $e_x,e'$ and $z$ all belong to the boundary of some Voronoi region $R(p)$; see Figure~\ref{fig:CommonAncestor}.

Let $\ell_{z,p}$ be the line through $z$ and $p$. 
By the definition of  $\mathcal{V}(P)$,  $\ell_{z,p}$ intersects the boundary of $R(p)$ only at the point $z$, and it
separates $x$ and $y$.  
Thus the intersection point $z'$ between $[x,y]$ and $\ell_{z,p}$ belongs to $R(p)$.
Recall that Proposition~\ref{CirculosEnSegmento} implies that $C(z')$ is a separating circle.
Trivially, $C(z)\subseteq C(z')$ since the latter is obtained by expanding $C(z)$ while anchoring it to $p$.
Since $C(z)$ is a $P$-circle contained inside a separating circle, by
Observation~\ref{Relation P-circle separating circle} $C(z)$ is also a separating circle. 
\end{proof}

\begin{figure}[h]
\begin{center}
\includegraphics[width=90mm]{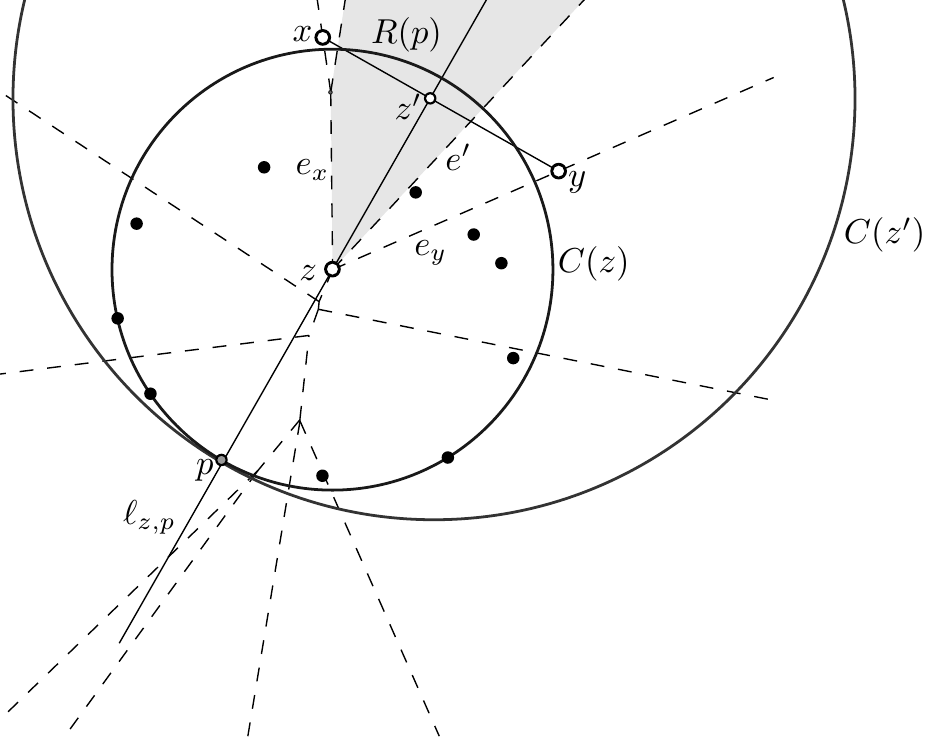}
\caption{\small Illustration for  Lemma~\ref{Common Ancestor}.} 
\label{fig:CommonAncestor}
\end{center}
\end{figure}

The following is also proved in~\cite{DynamicCircleSeparability}.
\begin{theorem}\label{sep in pi_s}
Let $s$ be a point on $\mathcal{V}(P)$. If $s$ is a separating point, 
then $\sep$  belongs to $\pi_s$.
\end{theorem}
\begin{proof}
Proceed by contradiction and assume that $\sep$ does not belong to $\pi_s$.
Let $z$ be the lowest common ancestor of $\sep$ and $s$. Note that $z$ lies in $\pi_s$ and therefore $z\neq \sep$.
In this case, Lemma~\ref{Common Ancestor}, in conjunction with Proposition~\ref{Rho is monotonic}, imply that $C(z)$ is a separating circle with $\rho(z)< \rho(\sep)$ which is a contradiction.
\end{proof}

Given a separating point $s$,
we claim that if we move a point $y$ continuously from $s$ towards $\cp$ on $\pi_s$, 
then $C(y)$ will shrink and approach $Q$, 
becoming tangent to it for the first time when $y$ reaches $\sep$. To prove this claim in Lemma~\ref{lemma:TangentToQ}, we introduce the following notation. 

Let $x$ be a point lying on an edge $e$ of $\mathcal{V}(P)$ such that $e$ lies on the bisector of $p,p'\in P$. Let $C^-(x)$ and $C^+(x)$ be the two closed convex regions obtained by splitting the disk $C(x)$ with the segment  $[p,p']$. Assume that $x$ is contained in $C^-(x)$.

\begin{observation}\label{obs:PushingCircles}
Let $x,y$ be two points lying on an edge $e$ of $\mathcal{V}(P)$.
If $\rho(x) > \rho(y)$, then $C^+(x)\subset C^+(y)$ and $C^-(y)\subset C^-(x)$.
\end{observation}

\begin{figure}[h!]
\begin{center}
\includegraphics[width=75mm]{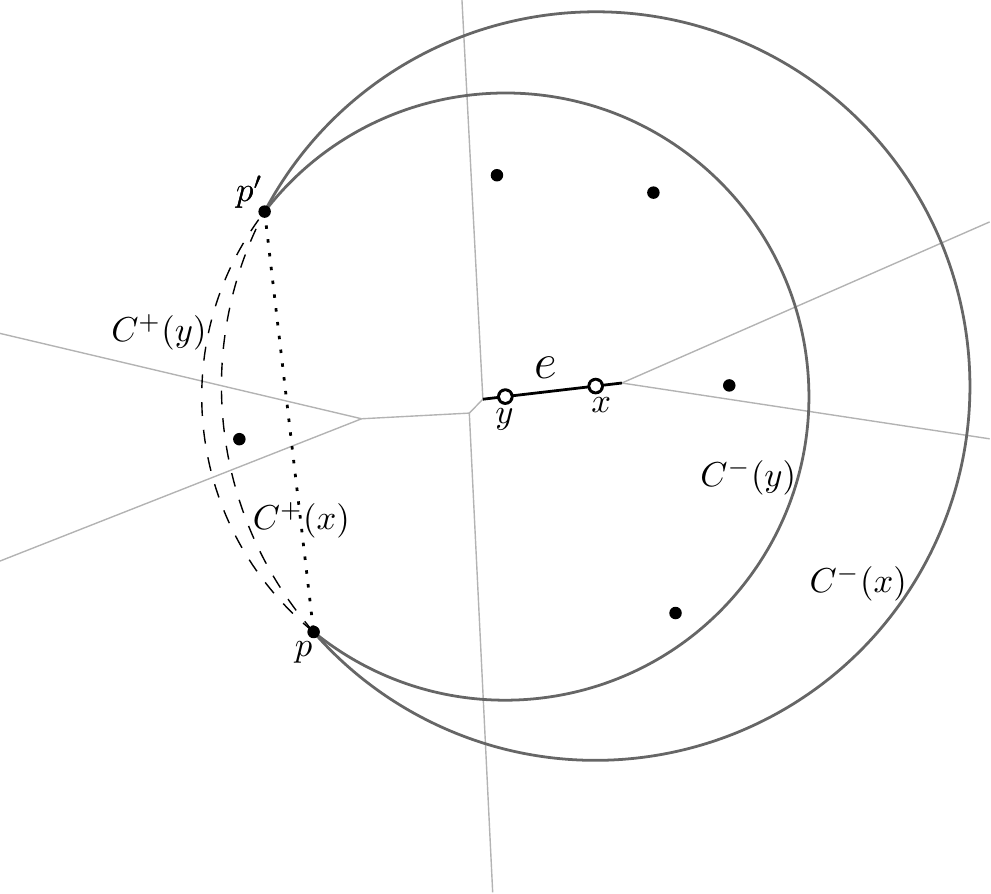}
\caption{\small The result of Observation~\ref{obs:PushingCircles} when $\rho(x) > \rho(y)$.} 
\label{fig:PushingCircles}
\end{center}
\end{figure}

Using this simple observation that can be seen in Figure~\ref{fig:PushingCircles}, we obtain the following generalization.

\begin{proposition}\label{prop:PushingCricles}
Let $s$ be a point on $\mathcal{V}(P)$ and let $x$ and $y$ be two points  on $\pi_s$. If $\rho(x) > \rho(y)$, then $C^+(x)\subset C^+(y)$ and $C^-(y)\subset C^-(x)$.
\end{proposition}
\begin{proof}
Note that if $x$ and $y$ lie on the same edge, then the result holds by Observation~\ref{obs:PushingCircles}. 
If they are on different edges, we consider the path $\Phi = (x,  v_0, \ldots,  v_k, y)$ contained in $\pi_s$ joining $x$ and $y$, such that $v_i$ is a vertex of $\mathcal{V}(P)$, $i\in \{0, \ldots, k\}$. 
Thus, Observation~\ref{obs:PushingCircles} and Proposition~\ref{Rho is monotonic} imply that $C^+(x)\subset C^+(v_0)\subset \ldots \subset C^+(v_k)\subset C^+(y)$ and that $C^-(y)\subset C^-(v_k)\subset \ldots \subset C^-(v_0)\subset C^-(x)$.
\end{proof}

Note that $\Sep = C(\sep)$ must intersect the boundary of $Q$. Otherwise, we could always push $\sep$ closer to the root on $\mathcal{V}(P)$, while keeping it as a separating point. 
Furthermore, since $Q$ is convex and $\Sep$ contains no point of $Q$ in its interior, 
the intersection consists on only one point.
From now on we refer to $\tang$ as the tangency point between $\Sep$ and $Q$.

We claim that $\tang$ lies on the boundary of $C^+(\sep)$. 
Assume to the opposite that $\tang$ lies on $C^-(\sep)$.  
Let $\varepsilon>0$ and let $c_\varepsilon$ be the point obtained by moving $\sep$ a $\varepsilon$ distance towards to $\cp$ on $\mathcal{V}(P)$. 
Note that by Proposition~\ref{Rho is monotonic}, $\rho(c_\varepsilon) < \rho(\sep)$. 
In addition, Proposition~\ref{prop:PushingCricles} implies that $C^-(c_\varepsilon)\subset C^-(\sep)$. 
Since we assumed that $\tang$ lies on the boundary of $C^-(\sep)$, we conclude that $\tang$ does not belong to $C(c_\varepsilon)$.
This implies that, for $\varepsilon$ sufficiently small, $C(c_\varepsilon)$ is a separating circle which is a contradiction to the minimality of $\Sep$.

\begin{lemma}\label{lemma:TangentToQ}
Let $s$ be a separating point. 
If $x$ is a point lying on $\pi_s$, then $C(x)$ is a separating circle if and only if $\rho(x) \geq \rho(\sep)$.
Moreover, $\Sep$ is the only separating circle that intersects $Q$.
\end{lemma}

\begin{proof}
Let $s$ be a separating point. We know by Theorem~\ref{sep in pi_s} that $\sep$ belongs to $\pi_s$.
Let $x_1$ and $x_2$ be two points on $\pi_s$ such that  $\rho(x_1) < \rho(\sep)$ and  $\rho(\sep) < \rho(x_2)$.
Proposition~\ref{prop:PushingCricles} implies that $C^+(\sep)\subset C^+(x_1)$ and since $\tang$ belongs to the boundary of $C^+(\sep)$, 
we conclude $C(x_1)$ contains $\tang$ in its interior. Therefore $C(x_1)$ is not a separating circle.

On the other hand, $C(x_2)$ contains no point of $Q$. Otherwise, let $q\in Q$ be a point lying in $C(x_2)$. 
Two cases arise:
Either $q$ belongs to $C^-(x_2)$ or $q$ belongs to $C^+(x_2)$. In the former case, since $\rho(s) > \rho(x_2)$, $q\in C^-(x_2)\subset C^-(s)$--- a contradiction since $C(s)$ is a separating circle. 
In the latter case, since $\rho(x_2) > \rho(\sep)$, Proposition~\ref{prop:PushingCricles} would imply that $q$ belongs to the interior of $\Sep$ which would also be a contradiction.
\end{proof}

\section{The algorithm}
The basis of our algorithm is to find a separating point $s$ and from there, perform a binary search on $\pi_s$ to find a separating circle tangent to $Q$ with center on this path.  
This is a rather intuitive solution, which was also used in~\cite{DynamicCircleSeparability}.  
However we use some additional geometric properties to reduce the time complexity.

\subsection{Preprocessing}
We first compute $\mathcal{V}(P)$ and $\cp$ in 
$O(n \log n)$ time~\cite{shamos}.
$\mathcal{V}(P)$ can be stored as a binary tree  with $n$ (unbounded) leaves, 
so that every edge and every vertex of the tree has a set of pointers to the vertices of $P$ defining it.
Every Voronoi region is stored as a convex polygon and
 every vertex $p$ of $P$ has a pointer to $R(p)$.
If $\cp$ is not a vertex of $\mathcal{V}(P)$, 
we split the edge that it belongs to.

We want our data structure to support binary search queries on any possible path $\pi_s$ of $\mathcal{V}(P)$.
Thus, to guide the binary search we would like to have an oracle that answers queries of the following form:
Given a vertex $v$ of $\pi_s$, decide if $\sep$ lies either between $\cp$ and $v$ or between $v$ and $s$ in $\pi_s$.
By Lemma~\ref{lemma:TangentToQ}, we only need to decide if $C(v)$ is a separating circle.

We will use an operation on the vertices of $\mathcal{V}(P)$ called $\textsc{FindPointBetween}$ with the following properties. 
Given two vertices $u,v$ in $\pi_s$, $\textsc{FindPointBetween}(u,v)$ returns a vertex $z$
that splits the path on $\pi_s$ joining $u$ and $v$ into two subpaths. 
Moreover, if we use our oracle to discard one of the subpaths and to proceed recursively on the other, then we want $\textsc{FindPointBetween}(u,v)$ to guarantee that this recursive process ends after $O(\log n)$ steps. That is, after $O(\log n)$ iterations, the search interval becomes only an edge of $\pi_s$ containing $\sep$.

A data structure build on top of $\mathcal{V}(P)$ that support this operation was presented in~\cite{ConstrainedMinimumEnclosingCircleWithCenterOnAQueryLineSegment}. This data structure can be constructed in $O(n)$ time and uses linear space by storing a constant number of pointers on each vertex of $\mathcal{V}(P)$.


\subsection{Searching for $\sep$ on the tree}
Recall that if $\CP$ is  a separating circle then it is a trivial solution.  
Since $Q$ is a convex $m$-gon, this can be checked easily in $O(\log m)$ time~\cite{ComputingExtremeDistancesBetweenConvexPolygons}. 
Thus we will assume that $\CP$ is not the minimum separating circle, 
which implies that $\CP$ intersects $Q$.

To determine the position of $\sep$ on $\mathcal{V}(P)$, we
first  find a separating point $s$ and then
 search for $\sep$ on $\pi_s$ using our data structure.

To find $s$, we construct a separating line $L$ between $P$ and $Q$. 
This can be done in $O(\log n + \log m)$ time~\cite{ComputingExtremeDistancesBetweenConvexPolygons}.
Let $p_{_L}$ be the point of $P$ closest  to $L$ and assume that no other point in $P$ lies at the same distance; otherwise rotate $L$ slightly.
Let $L_{\perp}$ be the perpendicular to $L$ that contains $p_{_L}$ and let $s$ be the intersection 
of $L_{\perp}$ with the boundary of $R(p_{_L})$; see Figure~\ref{fig:PrimerCirculoSeparador}.
We know that  $L_\perp$ intersects $R(p_{_L})$ because 
 $L$ can be considered as a $P$-circle, containing only $p_{_L}$, with center at infinity  on $L_\perp$.

\begin{figure}[h!]
\begin{center}
\includegraphics[width=65mm]{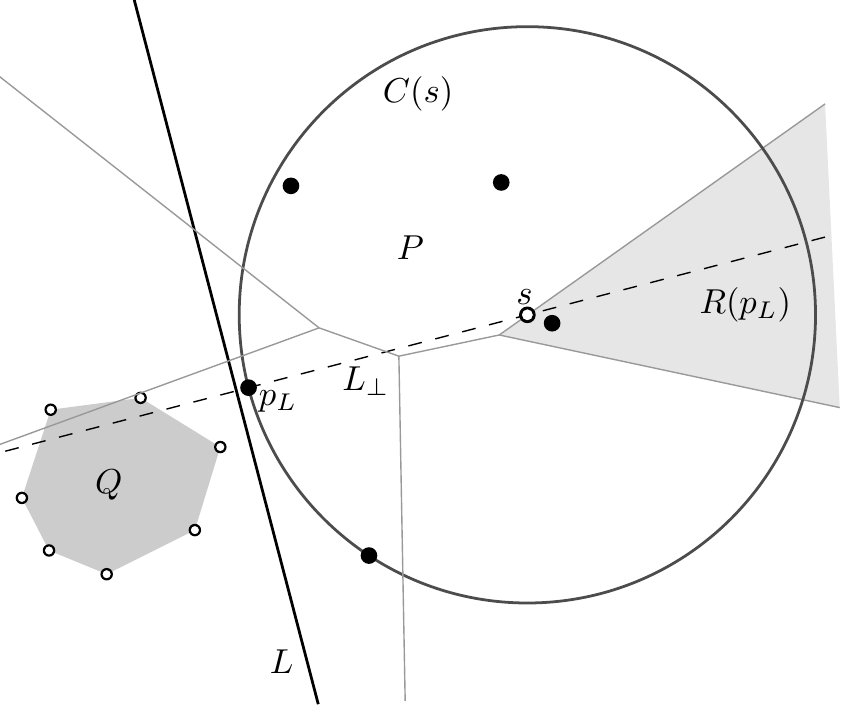}
\caption{\small Construction of $s$.} 
\label{fig:PrimerCirculoSeparador}
\end{center}
\end{figure}

Since $s$ is on the boundary of $R(p_{_L})$,  $C(s)$ passes through $p_{_L}$. Furthermore $C(s)$ is contained in the same halfplane defined by $L$ that contains $P$. 
So $C(s)$ is a separating circle.
Assume that $s$ lies on the edge $\overline{xy}$ of $\mathcal{V}(P)$ with $\rho(x) > \rho(y)$ and 
let $\pi_{s} = ( u_0 = s, u_1 = y, \ldots, u_r = \cp)$ be the path of length $r+1$ joining $s$ with $\cp$ in $\mathcal{V}(P)$. Theorem~\ref{sep in pi_s} implies that $\sep$ lies on $\pi_s$.\\

It is possible to use our data structure to perform a binary search on the vertices of $\pi_s$, computing, at each vertex $v$, the distance to $Q$ and the radius of $C(v)$. This way we could determine if $C(v)$ is a separating (or intersecting) circle.
However, this approach involves computing the distance to $Q$ at each step in $O(\log m)$ time, and thus takes  $O(\log n\cdot\log m)$ time.  This was the algorithm given in~\cite{DynamicCircleSeparability}.

It is worth noting that if the distance to our query object can be computed in $O(1)$ time, 
then the described algorithm takes $O(\log n)$ time. This is the case when instead of being an $m$-gon, $Q$ is either a circle or a point. The next result follows.

\begin{observation}\label{corollary:ResultForCircleAndPoints}
After preprocessing a set $P$ of $n$ points in $O(n \log n)$ time,
the minimum separating circle between $P$ and any given circle or point
can be found in $O(\log n)$ time.
\end{observation}

In the case of $Q$ being a convex $m$-gon, an improvement from the $O(\log n \log m)$ time algorithm can be obtained by strongly using the convexity of $Q$.

To determine if some point $v$ on $\pi_s$ is a separating point, it is not always necessary to
compute the distance between $v$ and $Q$.
One can first test, in constant time, if $C(v)$ intersects a  separating line tangent to $Q$.
If it does not intersect it, then $C(v)$ would be a separating circle and we can proceed with the binary search.
Otherwise, we can try to compute a new separating line tangent to $Q$ not intersecting $C(v)$.
The advantage of this is that while doing so, we reduce the portion of $Q$ that we need to consider in the future.
This is done as follows.


Compute the two internal tangents $L, L'$ between the convex hull of $P$ and $Q$ in $O(\log n + \log m)$ time. The techniques to construct these tangents are shown in Chapter 4 of~\cite{PreparataAndShamos}.

Let $q$ and $q'$ be the respective tangency points of $L$ and $L'$ with the boundary of $Q$. 
We will consider the clockwise polygonal chain $\varphi = [q = q_0, \ldots, q_k = q']$
 joining $q$ and $q'$ as in Figure~\ref{fig:InternalTangents}. 
 
 \begin{figure}[h!]
\begin{center}
\includegraphics{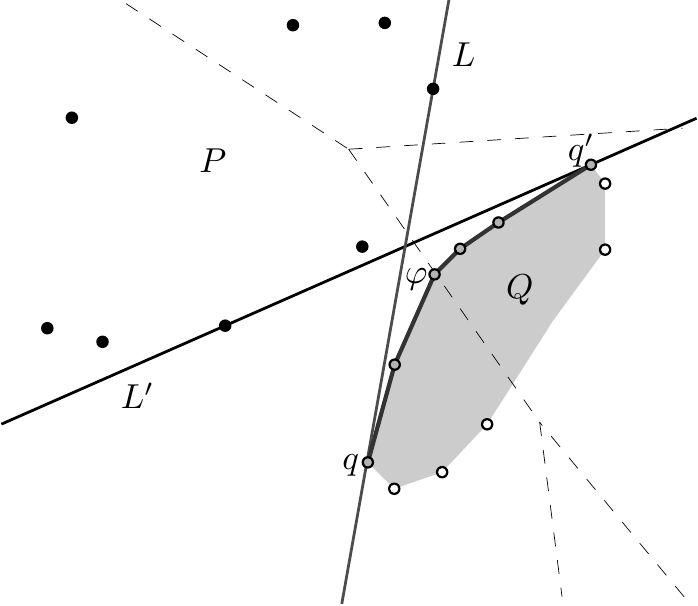}
\caption{\small The construction of $\varphi$.} 
\label{fig:InternalTangents}
\end{center}
\end{figure}

Recall that $\tang$ denotes the intersection point between $\Sep$ and the boundary of $Q$ and note that the tangent line to $\Sep$ at $\tang$ is a separating line.
Therefore, $\tang$ must lie on an edge of $\varphi$ since no separating line passes through any other boundary point of $Q$.

If $q = q'$, then $\tang = q$ and thus, we can forget about $Q$ and compute the minimum separating circle between $P$ and $q$.
As mentioned previously, by Observation~\ref{corollary:ResultForCircleAndPoints} this takes  $O(\log n)$ time.

Assume from now on that $q\neq q'$, as shown in  Figure~\ref{fig:InternalTangents}.
For each edge $e_i = q_iq_{i+1}$  ($0\leq i\leq k-1$) of $\varphi$, let $\ell_i$ be the line extending that edge. By construction, we know that each $\ell_i$ separates  $P$ and $Q$. We say that a point $x$ on $\ell_i$ but not in $e_i$ lies to the left of $e_i$ if it is closer to $q_i$, or to the right if it is closer to $q_{i+1}$.

Our algorithm will essentially perform two parallel binary searches, the first one on $\pi_s$ and the second one on $\varphi$, such that at each step we discard either a section of $\pi_s$ or a linear fraction of $\varphi$.   As we search on $\pi_s$, every time we find a separating circle, we move towards $\cp$.
When we confirm that a $P$-circle intersects $Q$, we move away from $\cp$.

As mentioned, we attempt to confirm if the current vertex $v$ being analyzed corresponds to a separating point. To do that, we compare $C(v)$ to some separating line $\ell_i$ for intersection, i.e., in constant time. 
If $C(v)$ is a separating circle, we instantly discard the section of the path lying below $v$ on $\mathcal{V}(P)$.  
If $C(v)$ does intersect $\ell_i$, we make a quick attempt to check if $C(v)$ intersects $Q$ by comparing $C(v)$ and the edge $e_i$ for intersection.  If so, $v$ is not a separating point and we can proceed with the binary search on $\pi_s$. Otherwise, the intersection of $C(v)$ with $\ell_i$ lies either to the left or to the right of $e_i$.
However, in this case we are not able to quickly conclude whether $C(v)$ intersects $Q$ or not.  Thus, we temporarily suspend the binary search on  $\mathcal{V}(P)$ and focus on $C(v)$, using it to eliminate half of $\varphi$.   
Specifically, the fact that $C(v)$ intersects $\ell_i$ to one side of $e_i$ (right or left) tells us that no future $P$-circle on our search will intersect  $\ell_i$ to the other side of $e_i$. 
This implicitly discards half of $\varphi$ from future consideration, and is discussed in more detail in the Theorem that follows.   
Thus, in constant time, we manage to remove a section of the path $\pi_s$, or half of $\varphi$, which will give the desired time bound.
The entire process is detailed in Algorithm~\ref{alg:DualBinarySearch}. 

\begin{algorithm}
\caption{Given $\pi_{s} = ( u_0 = s, u_1 = y, \ldots, u_r = \cp)$ and $\varphi = [q = q_0, \ldots, q_k = q']$,  find the edge of $\pi_s$ containing $\sep$}
\begin{algorithmic}[1]\label{alg:DualBinarySearch}
	\STATE Define the endpoints of the subpath of $\pi_s$ containing $\sep$, $u\gets s, v\gets \cp$
	\STATE Define the initial search interval on $\varphi$, $a\gets 0, b\gets k$
	\IF{$u$ and $v$ are consecutive vertices and $b= a+1$}\label{alg:MoveOn}
		\STATE End and return the segment $S=[u,v]$ and the segment $H = [q_a, q_b]$
	\ENDIF
	\STATE Let $z\gets \textsc{FindPointBetween}(u,v)$,  $j\gets \lfloor \frac{a+b}{2}\rfloor$ 	
	\STATE Let $e_j \gets \overline{q_j q_{j+1}}$ and let $\ell_j$ be the line extending $e_j$
	\IF{$b > a+1$} \label{alg:comparison}
		\STATE Compute $\rho(z)$ and 
		let $\delta \gets d(z, \ell_j)$, $\Delta\gets d(z, e_j)$
	\ELSE
		\STATE Compute $\rho(z)$ and let $\delta \gets d(z, e_j)$, $\Delta \gets d(z, e_j)$
	\ENDIF
	\IF{$\rho(z) \leq \delta$, that is $C(z)$ is a separating circle}
		\STATE Move forward on $\pi_s$, 
		$u\gets z$ and return to step~\ref{alg:MoveOn} \label{alg:redefine_u}
	\ELSE
		\IF{$\rho(z) > \Delta$ , that is if $C(z)$ is not a separating circle}
			\STATE Move backward on $\pi_s$, 
			$v\gets z$ and return to step~\ref{alg:MoveOn}  \label{alg:redefine_v}
		\ELSE
			\IF{$C(z)$ intersects $\ell_j$ to the left of $e_j$}
				\STATE We discard the polygonal chain to the right of $e_j$, 
				$b \gets \max\{j, a+1\}$  \label{alg:redefine_b}
			\ELSE
				\STATE We discard the polygonal chain to the left of $e_j$, 
				$a \gets j$   \label{alg:redefine_a}
			\ENDIF
			\STATE Return to step~\ref{alg:MoveOn}
		\ENDIF
	\ENDIF
\end{algorithmic}
\end{algorithm}

\begin{theorem}
Algorithm~\ref{alg:DualBinarySearch} finds the edge of $\pi_s$ containing $\sep$ in $O(\log n + \log m)$ time.
\end{theorem}
\begin{proof}
Our algorithm maintains two invariants.
The first  is that $C(u)$ is never a separating circle and $C(v)$ is always a separating circle. 
To begin with, $C(u) = C(s)$ is a separating circle, while
 $C(v) = \CP$ is not a separating circle. If either of these assumptions does not hold, the problem is solved trivially, without resorting to this algorithm.
Changes to $u$ and $v$ occur in steps~\ref{alg:redefine_u} or~\ref{alg:redefine_v}, and in both cases the invariant is preserved. Thus, $\sep$ always lies on the path joining $u$ with $v$.

The second invariant is that $\tang$, the tangency point between $\Sep$ and $Q$, 
always lies on the clockwise path joining $q_a$ with $q_b$ along $\varphi$. We already explained that the invariant holds when $a=0$ and $b=k$, corresponding to the inner tangents supporting $P$ and $Q$.
Thus we only need to look at steps~\ref{alg:redefine_b} and~\ref{alg:redefine_a}, where $a$ and $b$ are redefined. 

We will analyze step~\ref{alg:redefine_b}, however step~\ref{alg:redefine_a} is analogous. 
In step~\ref{alg:redefine_b} we know that $C(z)$ intersects $\ell_j$ to the left of $e_j$ and that $e_j$ does not intersect $C(z)$.
We claim that for every point $w$ lying on an edge of $\pi_s$, if $C(w)$ is a separating circle that intersects $\ell_j$, then it intersect it to the left of $e_j$. Note that if our claim is true, we can forget about the polygonal chain lying to the right of $e_j$ since no separating circle will intersect it.
To prove our claim, suppose that there is a point $w$ on $\pi_s$, such that $C(w)$ is a separating circle and $C(w)$ intersects $\ell_j$ to the right of $e_j$. 
Let $x$ and $x'$ be  two points on the intersection of $\ell_j$ with  $C(w)$ and $C(z)$, respectively.
Suppose first that $\rho(w)<\rho(z)$ and recall that by Proposition~\ref{prop:PushingCricles}, since $x'$ lies on $C^+(z)\subset C^+(w)$, $x'$ lies in $C(w)$. 
Thus, both $x$ and $x'$ belong to $C(w)$ which by convexity implies that $e_j$ is contained in $C(w)$. Therefore $C(w)$ is not a separating circle which is a contradiction.
Analogously, if $\rho(w) > \rho(z)$, then $e_j$ is contained in $C(z)$ which is directly a contradiction since we assumed the opposite; our claim holds.\\

Note that in each iteration of the algorithm, $a,b,u$ or $v$ are redefined so that either a linear fraction of $\varphi$ is discarded, or a part of $\pi_s$ is discarded and a new call to $\textsc{FindPointBetween}$ is performed. Recall that our data structure guarantees that $O(\log n)$ calls to $\textsc{FindPointBetween}$ are sufficient to reduce the search interval in $\pi_s$ to an edge~\cite{ConstrainedMinimumEnclosingCircleWithCenterOnAQueryLineSegment}. Thus, the algorithm finishes in $O(\log n + \log m)$ iterations. 

One extra detail needs to be considered when $b = a+1$. In this case only one edge $e = [q_a, q_{a+1}]$ remains from $\varphi$, and $\tang$ lies on that edge. Thus, if the line $\ell$ extending $e$ intersects $C(z)$ but $e$ does not, then either step \ref{alg:redefine_b} or \ref{alg:redefine_a} is executed. However, nothing will change in these steps and the algorithm will loop.
In order to avoid that, we check in step~\ref{alg:comparison} if only one edge $e$ of $\varphi$ remains. 
If this is the case, we know by our invariant that $\tang$ belongs to that edge and therefore we continue the search computing the distance to $e$ instead of computing the distance to the line extending it.
This way, the search on $\varphi$ stops but it continues on $\pi_s$ until the edge of $\mathcal{V}(P)$ containing $\sep$ is found.

Since we ensured  that every edge in $\mathcal{V}(P)$ has pointers to the points in $P$ that defined it, every step in the algorithm can be executed in $O(1)$ time. Thus, we conclude that Algorithm~\ref{alg:DualBinarySearch} finishes in $O(\log n + \log m)$ time. 

Since both invariants are preserved during the execution, Lemma~\ref{lemma:TangentToQ} implies that the algorithm returns segments $[u, v]$ from $\pi_s$ containing $\sep$, and $[q_a, q_{b}]$ from $\varphi$ containing $\tang$.
\end{proof}

From the output of Algorithm~\ref{alg:DualBinarySearch} it is trivial to obtain $\sep$ in constant time,
so we conclude the following.

\begin{corollary}
After preprocessing a set $P$ of $n$ points in $O(n \log n)$ time,
the minimum separating circle between $P$ and any query convex $m$-gon
can be found in $O(\log n + \log m)$ time.
\end{corollary}

\end{document}